\documentclass{aims}

\usepackage{txfonts}
\def\typeofarticle{Research article}

\numberwithin{equation}{section}

\makeatletter
\renewcommand{\@biblabel}[1]{#1\hfill \hspace{-0.2cm}}
\makeatother

\usepackage{amsmath,amssymb}
\usepackage{amsthm}

\newcommand{\x}{\mathbf{x}}
\newcommand{\y}{\mathbf{y}}
\newcommand{\z}{\mathbf{z}}
\newcommand{\bu}{\mathbf{u}}

\newcommand{\bP}{\mathbf{P}}
\newcommand{\cP}{\mathcal{P}}
\newcommand{\bT}{\mathbf{T}}
\newcommand{\bQ}{\mathbf{Q}}
\newcommand{\bI}{\mathbf{I}}
\newcommand{\bC}{\mathbf{C}}
\newcommand{\bD}{\mathbf{D}}
\newcommand{\bL}{\mathbf{L}}
\newcommand{\bU}{\mathbf{U}}
\newcommand{\bX}{\mathbf{X}}
\newcommand{\bY}{\mathbf{Y}}
\newcommand{\bA}{\mathbf{A}}
\newcommand{\bB}{\mathbf{B}}

\newtheorem{theorem}{Theorem}[section]
\newtheorem{lemma}[theorem]{Lemma}
\newtheorem{corollary}[theorem]{Corollary}
\newtheorem{example}[theorem]{Example}
\newtheorem{definition}{Definition}[section]

\emergencystretch=\maxdimen
\usepackage{cite}
\begin{document}

\title{A Note on Equivalent Conditions for Majorization}

\author{%
  Roberto Bruno
  and
  Ugo Vaccaro\corrauth
}

\shortauthors{the Author(s)}

\address{%
   {Department of Computer Science,
    University of Salerno, via Giovanni Paolo II, 132, 84084 Fisciano (SA), Italy}
  }

\corraddr{Email: uvaccaro@unisa.it; Tel: +39089969734;
 Fax: +39089969600.
}

\begin{abstract}
{
In this paper, we  introduced novel characterizations of
the classical concept of majorization in terms of upper triangular (resp., 
lower triangular) row-stochastic matrices, and in terms of
sequences of
linear transforms on 
vectors.
We used our new characterizations of 
 majorization 
to derive an improved entropy inequality.  }
\end{abstract}


\maketitle

\section{Introduction}

The concept of majorization has a rich history in mathematics with applications that span a wide range of disciplines. Majorization theory originated in economics
\cite{AS}, where it was employed to
rigorously explain the vague notion that the components of a given vector are
``more nearly equal" than the components of a different vector. Nowadays, majorization theory
finds applications in numerous areas, ranging from pure mathematics  to combinatorics
\cite{marshall1979inequalities,HA,Man},
from information
and communication theory \cite{M+,Mu, Sa1, Sa, cv, Wi, PY, JB, W} to thermodynamics 
and quantum theory
\cite{bb,S}, from mathematical chemistry \cite{bi} to optimization \cite{da}, and much more.

There are many equivalent conditions for majorization. We review the 
most common ones in Section~\ref{pre}.  Successively, 
we present our new conditions for majorization 
in Sections~\ref{sec:L} and \ref{sec:upper}. Finally, in Section~\ref{appl}, we present
an application to entropy inequalities.


\section{Preliminaries}\label{pre}
Throughout the paper, we will consider vectors $\x=(x_1, \ldots ,x_n)\in 
\mathbb{R}_+^n$ 
ordered component-wise, that is, for which $x_1\geq \ldots \geq x_n$.

\begin{definition}{\normalfont \cite[Def. A.1]{marshall1979inequalities}}
    For $\x,\y \in \mathbb{R}_+^n$,  
    \begin{align}\label{eq:x<y}
        \x \prec \y \mbox{ if }
        \begin{cases}
            \displaystyle \sum_{i=1}^k x_i \leq \sum_{i=1}^k y_i, \hspace{0.5cm} k=1,\ldots,n-1,\\ 
         \displaystyle\sum_{i=1}^n x_i = \sum_{i=1}^n y_i.    
        \end{cases}
    \end{align}
\end{definition}
\noindent
When $ \x \prec \y$, we say that vector $\x$ is \textit{majorized} by $\y$ 
(equivalently, that $\y$ majorizes $\x$). 
\smallskip

There are many equivalent conditions for majorization 
(e.g., see \cite{marshall1979inequalities}, Chapter 4). 
The conditions that are more closely related to the subject matter of our paper
are expressed in terms of doubly stochastic matrices and $T$-transforms.
\begin{definition}\label{doubly}
    An $n\times n$ matrix $\bP=[P_{ij}]$ is \textit{doubly stochastic} if
    \begin{equation*}
        P_{ij}\geq 0 \hspace{0.5cm} \forall i,j \in \{1,\dots,n\},
    \end{equation*}
    and 
    \begin{align*}
        \sum_{i=1}^n P_{ij} &= 1, \hspace{0.5cm} j=1,\dots,n;\\
        \sum_{j=1}^n P_{ij} &= 1, \hspace{0.5cm} i=1,\dots,n.
    \end{align*}
\end{definition}
\begin{definition}\label{T}
    An $n\times n$ matrix $\bT$ is  a \textit{{$\bT$}-transform} if
    \begin{equation*}
\bT = \lambda \bI + (1-\lambda)\bQ, 
    \end{equation*}
where $0\leq \lambda <1$, $\bI$ is the $n\times n$ identity matrix, and
$\bQ=[Q_{\ell m}]$ is a permutation matrix such that
\begin{equation}\label{Q}
    Q_{\ell m}=\begin{cases}
1 & \ \mbox{ for } \ell=m, \mbox{ and } \ell,m\notin\{i,j\} \\
1 &\ \mbox{ for } \ell=j \mbox{ and } m=i\\
1 & \ \mbox{ for } \ell=i \mbox{ and } m=j\\
0 & \ \mbox{otherwise,}
    \end{cases}
\end{equation}
for some indices $i,j\in \{1, \ldots , n\}$, $i\neq j$.
\end{definition}

Thus, for an arbitrary   $\x=(x_1, \ldots x_n)\in \mathbb{R}_+^n$,  it holds that the vector $\x\bT$ has the form
\begin{align}
    \x\bT = (x_1,\dots,x_{i-1}, \lambda x_i + (1-\lambda)x_j,x_{i+1},\dots, x_{j-1},
    \lambda x_j + (1-\lambda)x_i,x_{j+1},\dots,x_n).\label{Tform}
\end{align}
Notice that each $\bT$-transform is a doubly stochastic matrix.
It holds that:
\begin{theorem}{\normalfont\cite[Proposition B.6, Ch. 2]{marshall1979inequalities}}
    \label{th:equivalence_T}
    For any  \ $\x,\y \in \mathbb{R}^n$, the following conditions are equivalent:
    \begin{enumerate}
        \item[(1)] $\x \prec \y$;
        \item[(2)] $\x = \y\bP$ for some doubly stochastic matrix $\bP$;
        \item[(3)] $\x$ can be derived from $\y$ by successive applications of at most $n-1$
 $\bT$-transforms,  as described in Definition \ref{T}.
    \end{enumerate}
\end{theorem}

\section{Majorization by lower triangular stochastic matrices}\label{sec:L}
We start this section by introducing the concept of $\bA$-transform.

Informally, an $\bA$-transform of a vector $\x=(x_1, \ldots , x_n)$
is a transformation that involves two vector components, 
$x_i$ and $x_j$, with $i < j$. The transformation
operates on the vector $\x$ by increasing 
the value of the component $x_i$ {by} the quantity $\lambda x_j$ and decreasing the value of the component $x_j$ {by} the same value {$\lambda x_j$}, where $\lambda$ is
a real number 
$\lambda \in [0,1]$.
More formally,  an $\bA$-transform  can be described by the matrix:
\begin{equation}\label{I+X}
    \bA = \bI+ \bX,
\end{equation}
where $\bI$ is the $n\times n$ identity matrix and $\bX=[X_{\ell m}]$ is a matrix with all entries  equal to $0$
except for   two elements $X_{ji}$ and $X_{jj}$, for a given pair 
of indices $i,j$, $j>i$, where 
$X_{ji} = \lambda$ and $X_{jj} = -\lambda$. Thus, 
the vector $\x\bA$ has the form 
\begin{align}
    \x\bA = (x_1,\dots,x_{i-1},x_i + \lambda\,x_j,x_{i+1},\dots,x_{j-1}, 
    x_j - \lambda\,x_j,x_{j+1},\dots,x_n)\label{L1}.
\end{align}

Note that the matrix $\bA=[L_{\ell m}]$ is lower triangular and  row-stochastic, that is,

\begin{align}\label{eq:L}
    &A_{\ell m}\geq 0 \ \mbox{ for each } \ell,m,\\
    &A_{\ell m}=0 \ \mbox{ for } \ell {<} m, \\
    &\sum_{m=1}^n A_{\ell m}=1 \ \mbox{ for all } \ell.
\end{align}
The following theorem holds.

\begin{theorem}
        \label{th:l_transforms}
    Let $\x, \y \in \mathbb{R}_+^n$. 
    It holds that $\x \prec \y$ if, and only if, 
     $\y$ can be derived from $\x$ by the successive applications of a finite number of $\bA$-transforms.
\end{theorem}
\begin{proof}
    
    Let $\x\prec \y$. To avoid trivialities, we assume $\x\neq \y$. 
    We shall prove that
    $\y$ can be derived from 
    $\x$ by the successive applications of a finite number 
    of $\bA$-transforms.
    
    Since the first condition  of (\ref{eq:x<y}) holds, 
    there is an index $j$ such that
    \begin{equation}\label{eq:j<k}
\sum_{i=1}^{j-1}x_i=\sum_{i=1}^{j-1}y_i \ \  \mbox{ and } \ \ 
\sum_{i=1}^{j}x_i<\sum_{i=1}^{j}y_i. 
  \end{equation}
For such an index $j$, it holds that $x_j<y_j$.
From (\ref{eq:j<k}) and the second condition of (\ref{eq:x<y}), 
we get that there exists an index $k>j$ 
such that $x_k>y_k.$
    
Let $j$ be the \textit{smallest} index such that $y_j > x_j$, and let $k$ be the \textit{smallest} index greater than $j$ such that $x_k > y_k$. 

Let 
\begin{equation}\label{delta}
\delta = \min{\{}y_j - x_j, x_k - y_k{\}}.
\end{equation}
We define an $\bA$-transform as in (\ref{L1}), with
$\lambda= {\delta}/{x_k}$ and $\bX=[X_{\ell m}]$ defined as follows:
 \begin{equation*}
        X_{\ell m} =\begin{cases}
            \lambda &\qquad\mbox{ if } \ell=k \mbox{ and }  m=j,\\
            -\lambda &\qquad\mbox{ if }  \ell=m=k,\\
            0 &\qquad\mbox{ otherwise. }
        \end{cases}
    \end{equation*}
The application of such a matrix $\bA$ on the vector $\x$
gives the vector $\x\bA=\x'$ with components
\begin{align}
    \x' = (x_1, \dots, x_{j-1},x_j + \delta, x_{j+1},\dots, x_{k-1},
   x_k - \delta,x_{k+1},\dots,x_n).\label{x'}
\end{align}

{
We pause here to illustrate the rest of 
our proof technique, which proceeds through the following steps:\\
(1) We compute the smallest index
$j$ for which (\ref{eq:j<k}) holds. This means that vectors 
$\x$ and $\y$ coincide on the first $j-1$ components.\\
(2) We modify vector $\x$ according to the $\bA$  operator defined above, 
to get vector $\x'$ as described in~(\ref{x'}).\\
(3) We prove (\ref{eq:partial_majorization}) below, without altering the order of the components of 
$\x'=\x\bA$ (this is crucial).\\
(4) The number of components on which $\x'$ and $\y$ now coincide is greater than 
the number of components on which $\x$ and $\y$  coincide.
}

Let us show that the new vector $\x'$ satisfies the following property:
\begin{equation}
    \label{eq:partial_majorization}
    \begin{cases}
            \displaystyle \sum_{i=1}^\ell x'_i \leq \sum_{i=1}^\ell y_i,  \ell=1,\ldots,n-1,\\ 
         \displaystyle\sum_{i=1}^n x'_i = \sum_{i=1}^n y_i.    
        \end{cases}
\end{equation}

From (\ref{x'}) and since the vector $\x$ satisfies the first condition of (\ref{eq:x<y}), we get
\begin{equation}
    \label{eq:first_part}
    \sum_{i=1}^\ell x'_i =\sum_{i=1}^\ell x_i\leq \sum_{i=1}^\ell y_i,\hspace{0.5cm} \ell=1,\dots,j-1.
\end{equation}

From (\ref{delta}), we know that $x_j+\delta\leq y_j$. Thus,
from (\ref{x'}) and (\ref{eq:first_part}), we get
\begin{equation}\label{secondopasso}
     \sum_{i=1}^j x'_i =\sum_{i=1}^{j-1} x_i
     +(x_j+\delta)
     \leq \sum_{i=1}^{j-1} y_i+y_j=\sum_{i=1}^{j} y_i.
\end{equation}
By definition, 
    the index $k$ is the \textit{smallest} index greater than  $j$ 
    for which $x_k > y_k$. It follows 
    that 
    \begin{equation}
        \label{eq:greater_values}
        x_{\ell}\leq y_{\ell},\hspace{1cm}\ell=j+1,\dots,k-1.
    \end{equation}
Therefore, from (\ref{x'}) and (\ref{secondopasso}), we
obtain that for each  $ \ell=j+1,\dots, k-1$, it holds that
\begin{align}
     \sum_{i=1}^{\ell} x'_i &=
     \sum_{i=1}^{j-1} x_i  +(x_j+\delta)
                    +\sum_{i=j+1}^\ell x_j\nonumber\\
        & \leq\sum_{i=1}^{j-1} y_i+
        y_j+\sum_{i=j+1}^{\ell} x_i\nonumber\mbox{(since $x_j+\delta\leq y_j$) }\\ 
        & \leq\sum_{i=1}^{j-1} y_i+
        y_j+\sum_{i=j+1}^{\ell} y_i\hspace{0.5cm}\mbox{(from (\ref{eq:greater_values}))}\nonumber\\
        & = \sum_{i=1}^{\ell} y_i.\label{terzopasso}
\end{align}
From (\ref{x'}) and since the vector $\x$ satisfies the first condition of (\ref{eq:x<y}), we get
\begin{align}
    \sum_{i=1}^k x'_i& = \sum_{i=1, i\neq j,k}^k x_i + 
    (x_j+\delta)+(x_k-\delta)\nonumber\\
    &= \sum_{i=1}^k x_i 
    \leq \sum_{i=1}^k y_i.\label{quartopasso}
\end{align}

Finally, since the vector $\x$ satisfies the first and second condition of (\ref{eq:x<y}), we have that
\begin{equation}
    \label{quintopasso}
    \begin{cases}
            \displaystyle \sum_{i=1}^\ell x'_i = \sum_{i=1}^{\ell} x_i \leq \sum_{i=1}^\ell y_i, \hspace{0.5cm} \ell=k+1,\ldots,n-1,\\ 
         \displaystyle\sum_{i=1}^n x'_i = \sum_{i=1}^{n} x_i = \sum_{i=1}^n y_i.    
        \end{cases}
\end{equation}
Therefore, from (\ref{eq:first_part}), (\ref{secondopasso}), and (\ref{terzopasso})--(\ref{quintopasso}), we have that (\ref{eq:partial_majorization}) holds.

Notice that if $\delta = y_j - x_j$, then $x'_j$ is equal to $y_j$; equivalently,  if  $\delta = x_k - y_k$, then $x'_k$ will be equal to $y_k$. Thus, the vector $\x'=\x\bA$ has at least one additional component (with respect to $\x$) that is equal to a component of $\y$.
Moreover, 
since each $\bA$-transform preserves the property (\ref{eq:partial_majorization}), we can iterate the process
starting from $\x'=\x\bA$.
It follows that  $\y$ can be derived from $\x$ by the application of a finite number of $\bA$-transforms.

\smallskip
Let us prove the converse part of the theorem.
Hence, we assume that  $\y$ can be derived from $\x$ by a finite number of $\bA$-transforms, and we prove
 that $\x \prec \y$.
Let 
\begin{equation}\label{adj}
\x_1=\x\bA_1, \ \x_2=\x_1\bA_2,\dots, \x_k=\x_{k-1}\bA_k = \y
\end{equation}
be the vectors obtained by the successive applications of a number $k$ of $\bA$-transforms.   
Given an arbitrary vector $\z\in \mathbb{R}_+^n$, let us denote  with $\z^{\downarrow}$ the vector with the same
components as $\z$, ordered in a nonincreasing fashion. 
From the definition (\ref{x'}) of  $\bA$-transform,
it follows that the partial sums of $\x'$ are certainly greater than {or equal to} the
corresponding partial sums of $\x$. Therefore, 
$$\x \prec \x_1^{\downarrow} \prec \x_2^{\downarrow} \prec \dots \prec \x_k^{\downarrow}=\y.\footnote{Recall
that $\x=\x^{\downarrow}$ and $\y=\y^{\downarrow}$.}$$
By the transitivity property of the partial order relation $\prec$, we get  $\x \prec \y$. 
\end{proof}

\begin{corollary}
    \label{lemma:n_minus_1}
    Let $\x, \y \in \mathbb{R}_+^n$. If $\x \prec \y$, then $\y$ can be derived from $\x$ by the successive application of, at most, $n-1$ $\bA$-transforms.
\end{corollary}
\begin{proof}
    In the proof of Theorem \ref{th:l_transforms}, we have shown 
    that the application of each $\bA$-transform 
    equalizes \textit{at least} one component 
    of the intermediate vectors $\x_j=\x_{j-1}\bA$ 
    to a component of $\y$. Since all vectors appearing in (\ref{adj}) have an equal sum, the last $\bA$-transform always equalizes both the affected components to the respective components of $\y$.  As a result, $\y$ can be obtained  by the application of
     at most $n-1$ \ $\bA$-transforms.
\end{proof}

Although it is evident, let us explicitly mention that if $\y$ can be derived by the application of
     at most $n-1$ \ $\bA$-transforms, then 
     it holds that $\y=\x\bL$, where the matrix
     $\bL$ is the product of the individual $\bA$-transforms.
     
The following technical lemma is probably already known in the literature. 
Since we have not found a  source that explicitly
mention{s} it, we provide a proof of it to maintain the paper self-contained.

\begin{lemma}
    \label{lemma:matrix_mul}
    Let $\bC$ and $\bD$ be two $n\times n$ lower triangular row-stochastic matrices. The product matrix $\bC\bD$ is still a lower triangular row-stochastic matrix.
\end{lemma}
\begin{proof}
    Since $\bC\bD$ is the product of two lower triangular matrices, 
    one can see that $\bC\bD$ is a lower triangular matrix too. Thus, we need only to show that it is row-stochastic. 

    First of all, each entry  $(CD)_{ij}$ of $\bC\bD$ is nonnegative
    since it is 
    the sum of nonnegative values. Let us consider 
    the sum of the elements of the $i$-th row:
    \begin{align*}
        \sum_{j=1}^n (CD)_{ij} =& \sum_{j=1}^n \left(\sum_{k=1}^n C_{ik}D_{kj}\right)\\
        =& \sum_{k=1}^n C_{ik} \left(\sum_{j=1}^n D_{kj}\right)\\
        =& \sum_{k=1}^n C_{ik}\cdot 1 \hspace{1cm} \mbox{(since $\sum_{j=1}^n D_{kj} = 1$)}\\
        =& 1 \hspace{2.7cm} \mbox{(since $\sum_{k=1}^n C_{ik} = 1$)}.
    \end{align*}
    Thus, since the above reasoning holds for each $i=1,\dots,n$, the matrix $\bC\bD$ is a lower triangular row-stochastic matrix.
\end{proof}
{It may be worth commenting that Lemma \ref{lemma:matrix_mul} also holds for a product of column-stochastic matrices, which gives a column-stochastic matrix. This holds since $(\bC\bD)^T=\bD^T\bC^T$, where, for an arbitrary matrix $\bD$, $\bD^T$ denotes the transpose of $\bD$, and the transpose of a row-stochastic matrix is a column-stochastic matrix.}

{The} next Theorem characterizes majorization in terms of lower triangular row-stochastic matrices.

\begin{theorem}
    \label{th:lower_triangular}
    Let $\x,\y \in \mathbb{R}_+^n$. It holds that $\x \prec \y$ if, and only if, there exists a lower triangular row-stochastic matrix $\bL$ such that $\y = \x\bL$.
\end{theorem}
\begin{proof}
    The implication that if $\x \prec \y$, then there exists a lower triangular row-stochastic matrix $\bL$ such that $\y = \x\bL$ directly {follows from} the results of Theorem \ref{th:l_transforms} and Lemma \ref{lemma:matrix_mul}.
    Indeed, from Theorem~\ref{th:l_transforms} we know that $\y$ can be obtained from $\x$ by the successive application of $\bA$-transforms, and from Lemma \ref{lemma:matrix_mul} we know that the product of two consecutive $\bA$-transforms is still a lower triangular row-stochastic matrix. Hence, the matrix $\bL$ obtained as the product of all the $\bA$-transforms is a lower triangular row-stochastic matrix for which $\y = \x\bL$.

{
    We notice that a different proof of the 
    above result has been given by Li in \cite[Lemma 1]{Li} by an application of 
    Algorithm 1 in \cite{CGV}
    }

    \smallskip
    
    Let us now prove the converse implication.
    Assume that there exists a lower triangular row-stochastic matrix $\bL=[L_{ij}]$ for which $\y = \x \bL$. Let us prove that $\x \prec \y$. For such a purpose, we can express each component of $\y$ in the following way:
    \begin{equation}
    \label{eq:y}
    y_i = \sum_{j=i}^n L_{ji} x_j,\hspace{0.5cm} i=1,\dots,n.
    \end{equation}

Hence, by using (\ref{eq:y}), we can rewrite the sum of the first $k$ components of $\y$ as follows:
$$
   y_1 + \dots + y_k =  \sum_{j=1}^n L_{j1}x_j + \sum_{j=2}^n L_{j2}x_j + \dots + \sum_{j=k}^n L_{jk}x_j.
$$
By grouping the multiplicative coefficients of each $x_i$, we get:
\begin{equation}
    \label{eq:partial_sum}
   y_1 + \dots + y_k =  x_1\sum_{i=1}^1 L_{1i} + x_2\sum_{i=1}^2 L_{2i} + \dots + x_k\sum_{i=1}^k L_{ki} + \dots + x_n \sum_{i=1}^k L_{ni}.
\end{equation}

Since the matrix $\bL$ is lower triangular row-stochastic, we have that for each $j=1,\dots,n$, it holds that $\sum_{i=1}^j L_{ji} = 1$. Hence, from (\ref{eq:partial_sum}), we get
$$   y_1 + \dots + y_k \geq  x_1\,\sum_{i=1}^1 L_{1i} + x_2\,\sum_{i=1}^2 L_{2i} + \dots + x_k\,\sum_{i=1}^k L_{ki} = x_1 + \dots + x_k$$
for each $k=1,\dots,n${, and $$
y_1 + \dots + y_n=x_1 + \dots+x_n.
$$} 
Thus, $\x \prec \y$.
\end{proof}

We point out the following interesting property of
the matrix $\bL$ that appears in the
``if" part of Theorem \ref{th:lower_triangular}.
\begin{corollary}
Let $\x, \y \in \mathbb{R}_+^n$. If \ $\x \prec \y$, then 
there exists a lower triangular row-stochastic matrix  $\bL$,
with at most $2n-1$ nonzero elements, such that  $\y=\x\bL$.
\end{corollary}

\begin{proof}
    From Corollary \ref{lemma:n_minus_1} and Theorem \ref{th:lower_triangular}, we know that the matrix $\bL$, 
    for which it holds that $\y=\x\bL$, is  the product of at most $n-1$ $\bA$-transforms.
    Let 
    \begin{equation}\label{eq:sequence_A}
        \bA_1,\dots,\bA_t
    \end{equation}
    be the individual  matrices associated with such  $t$ $\bA$-transforms, $t<n$.
By Definition (\ref{I+X}), the matrix $\bA_1$ has $n+1$ nonzero elements.
    
Let $\bC$ be the matrix equal to the product of the first $m-1$ $\bA$-transforms of (\ref{eq:sequence_A}), $m\leq t <n$, that is, $\bC=\prod_{i=1}^{m-1}\bA_i$.
We show that the product
    \begin{equation}\label{eq:product}
        \bC \bA_{m}  
    \end{equation}
gives a matrix with only one additional nonzero element with respect to $\bC$.

    Let $j, k$ be the pair of {indices} chosen to construct $\bA_{m} = \bI + \bX_{m}$. From (\ref{eq:product}), we get
    \begin{equation}\label{eq:product_ext}
        \bC \bA_{m} = \bC (\bI +\bX_{m}) = \bC + \bC  \bX_{m}.
    \end{equation}
    
    For every $\bA$-transform, we recall that 
    we always choose the \textit{smallest} index $j$ such that $y_j > x_j$, and the \textit{smallest} index $k$ greater than $j$ for which $y_k < x_k$. 
    Therefore, all the previous $\bA$-transforms have chosen
    {indices} less than or at most equal to $k$. Consequently, in the matrix $\bC$, all the rows after the $k$-th row are equal to the rows of the identity matrix. 
    By the definition, the matrix $\bX_{m}$ has nonzero elements only in the $k$-th row, 
    (in positions $X_{kj}$ and $X_{kk}$, respectively).
    Hence, the matrix $\bC \bX_{m}$ 
    has nonzero elements only in the entries  $(C  X_{m})_{kj}$ and $(C X_{m})_{kk}$. Since the element $(C X_{m})_{kk}$ will be added to the diagonal element of $\bC$, the \textit{only new} nonzero  element in $\bC \bX_{m}$ (with respect to $\bC$) is $(CX_{m})_{kj}$. 
    Hence, from (\ref{eq:product_ext}), we get that the product $\bC\bA_m$ gives a matrix with only one new element with respect to $\bC$. 
    
    Since each product generates a matrix with only one additional nonzero element with respect to the previous one, 
    we obtain that the final matrix $\bL$ has at most $n+1 + n - 2 = 2n -1$ nonzero elements.
\end{proof}

We  summarize our results in the next theorem, mirroring
the classic Theorem \ref{th:equivalence_T}.
\begin{theorem}
    If $\x, \y \in \mathbb{R}_+^n$, the following conditions are equivalent:
    \begin{enumerate}
        \item[(1)] $\x \prec \y$;
        \item[(2)] $\y = \x\bL$ for some lower triangular row-stochastic matrix $\bL$;
        \item[(3)] $\y$ can be derived from $\x$ by successive applications of at most $n-1$
 $\bA$-transforms,  as defined  in~(\ref{L1}).
    \end{enumerate}
\end{theorem}
\begin{proof}
    The equivalences {follow} from the results of Theorems \ref{th:l_transforms} and \ref{th:lower_triangular}.
\end{proof}

Let us look at an example of how the matrix $\bL$ is constructed.

\begin{example}
Let $\y = [0.6, 0.2, 0.1, 0.1]$ and $\x = [0.3, 0.3, 0.3, 0.1]$ with $\x \prec \y$.

The first $\bA$-transform affects the elements $x_1$ and $x_2$ with $\delta = \min{\{}y_1-x_1,x_2-y_2{\}} = 0.1$ and $\lambda={\delta}/{x_2} = {1}/{3}$. As a result, the associated  matrix $\bA_1$ is:
$$
    \bA_1 = \begin{bmatrix}
            1 & 0 & 0 & 0\\
            0 & 1 & 0 & 0\\
            0 & 0 & 1 & 0\\
            0 & 0 & 0 & 1
        \end{bmatrix} +
        \begin{bmatrix}
            1 & 0 & 0 & 0\\
            \frac{1}{3} & -\frac{1}{3} & 0 & 0\\
            0 & 0 & 1 & 0\\
            0 & 0 & 0 & 1
        \end{bmatrix} =
        \begin{bmatrix}
            1 & 0 & 0 & 0\\
            \frac{1}{3} & \frac{2}{3} & 0 & 0\\
            0 & 0 & 1 & 0\\
            0 & 0 & 0 & 1
        \end{bmatrix}
$$
and $\x^1 {= \x \bA_1} = [0.4, 0.2, 0.3, 0.1]$.

The second $\bA$-transform affects the elements $x^1_1$ and $x^1_3$ with $\delta= \min{\{}y_1-x^1_1, x^1_3-y_3{\}}= 0.2$ and $\lambda={\delta}/{x^1_3}={2}/{3}$. Hence,
the matrix $\bA_2$ is 
$$
    \bA_2 = \begin{bmatrix}
            1 & 0 & 0 & 0\\
            0 & 1 & 0 & 0\\
            0 & 0 & 1 & 0\\
            0 & 0 & 0 & 1
        \end{bmatrix} +
        \begin{bmatrix}
            1 & 0 & 0 & 0\\
            0 & 1 & 0 & 0\\
            \frac{2}{3} & 0 & -\frac{2}{3} & 0\\
            0 & 0 & 0 & 1
        \end{bmatrix} =
        \begin{bmatrix}
            1 & 0 & 0 & 0\\
            0 & 1 & 0 & 0\\
            \frac{2}{3} & 0 & \frac{1}{3} & 0\\
            0 & 0 & 0 & 1
        \end{bmatrix}
$$
and $\x^2 {= \x^1 \bA_2}= [0.6, 0.2, 0.1, 0.1] = \y$. Therefore, the final matrix for which $\y=\x\bL$ is
\begin{equation}
    \label{L_matrix_ex}
    \begin{aligned}
        \bL = \bA_1 \bA_2 =
        \begin{bmatrix}
            1 & 0 & 0 & 0\\
            \frac{1}{3} & \frac{2}{3} & 0 & 0\\
            \frac{2}{3} & 0 & \frac{1}{3} & 0\\
            0 & 0 & 0 & 1
        \end{bmatrix}.
    \end{aligned}
\end{equation}
It is worth noticing that the matrix $\bL$ is not the inverse of the {doubly stochastic matrix} $\bP$, for which $\x=\y\bP$, obtained by applying a series of  $\bT$-transforms (Theorem \ref{th:equivalence_T}). In fact, the inverse of $\bL$ is 
$$
\begin{bmatrix}
            1 & 0 & 0 & 0\\
            -\frac{1}{2} & \frac{3}{2} & 0 & 0\\
            -2 & 0 & 3 & 0\\
            0 & 0 & 0 & 1
        \end{bmatrix}
$$
which is not a doubly stochastic matrix.
\end{example}


\section{Majorization by upper triangular stochastic matrices}\label{sec:upper}

We now present
our characterization of majorization through upper triangular row-stochastic matrices.

A $\bB$-\textit{transformation} or, more briefly, a $\bB$-\textit{transform}, is a transformation of a vector $\y =(y_1,\dots,y_n)$ that involves two vector components, $y_i$ and $y_j$, with $i<j$. The transformation operates on the vector $\y$ by decreasing the component $y_i$ {by} the quantity $\lambda y_i$ and increasing the component $y_j$ {by} the same value {$\lambda y_i$}, where $\lambda$ is a real number such that $\lambda \in [0,1]$. 

We can describe a $\bB$-transform as a matrix
\begin{equation}\label{defB}
    \bB = \bI + \bY,
\end{equation}
where $\bI$ is the identity matrix and $\bY=[Y_{\ell m}]$ is a matrix with all entries equal to 0 except for two elements $Y_{ii}$ and $Y_{ij}$, where  $Y_{ii} = -\lambda$ and $Y_{ij} = \lambda$. Thus, $\y\bB$ has the form
\begin{align}\label{u_transform}
    \y\bB = (y_1,\dots,y_{i-1}, y_i - \lambda y_i,y_{i+1},\dots,y_{j-1},
    y_j + \lambda y_i, y_{j+1},\dots,y_n).
\end{align}
Note that the matrix $\bB=[B_{\ell m}]$ is upper triangular and  row-stochastic, that is,
\begin{align}\label{eq:U}
    &B_{\ell m}\geq 0 \ \mbox{ for each } \ell,m,\\
    &B_{\ell m}=0 \ \mbox{ for } \ell{>}m, \\
    &\sum_{m=1}^n B_{\ell m}=1 \ \mbox{ for all } \ell.
\end{align}

The next theorem 
relates the $\bB$-transforms to majorization.

\begin{theorem}
    \label{th:u_transforms}
    Let $\x,\y\in \mathbb{R}_+^n$. It holds that $\x \prec \y$ if, and only if, $\x$ can be derived from $\y$ by the successive applications of a finite number of $\bB$-transforms.
\end{theorem}
\begin{proof}
    Let 
    $\x=(x_1, \ldots , x_n)\prec\y=(y_1, \ldots , y_n)$, $\x \neq \y$.
    We shall prove that  $\x$ can be derived from $\y$ by the successive applications of a finite number of $\bB$-transforms.

   Let $j$ be the \textit{largest} index for which $y_j > x_j$, and let $k$ be the \textit{smallest} index greater than $j$ such that $x_k > y_k$.
    Note that such a pair $j, k$ must exist (as we argued in the proof of Theorem \ref{th:l_transforms}).
    We set the quantity $\delta$~as
    \begin{equation}
        \label{eq:deltaB}    
        \delta = \min{\{}y_j - x_j, x_k-y_k{\}}.
    \end{equation}
    
    We define a $\bB$-transform as in (\ref{defB}), with $\lambda = \delta/y_j$ and $\bY$ such that 
    \begin{equation*}
        Y_{\ell m} =\begin{cases}
            -\lambda &\qquad\mbox{ if } \ell=m=j,\\
            \lambda &\qquad\mbox{ if } \ell=j \mbox{ and } m=k,\\
            0 &\qquad\mbox{ otherwise. }
        \end{cases}
    \end{equation*}
    The application of such a matrix $\bB$ on the vector $\y$ gives the vector $\y\bB=\y'$ with components
    \begin{align}\label{eq:y'}
        \y' = (y_1,\dots,y_{j-1},y_j -\delta, y_{j+1},\dots,y_{k-1},y_k +\delta, y_{k+1},\dots,y_n).
    \end{align}
    Let us show that the new vector $\y'$ still majorizes $\x$.
    
    From (\ref{eq:y'}) and since $\x \prec \y$, we get
    \begin{equation}\label{eq:passouno_B}
        \sum_{i=1}^{\ell} x_i \leq \sum_{i=1}^{\ell} y_i = \sum_{i=1}^{\ell} y'_i, \hspace{0.5cm}\ell=1,\dots,j-1.
    \end{equation}
    From (\ref{eq:deltaB}) {and (\ref{eq:y'})}, we know that 
    $$
        y'_j \geq x_j.
    $$
    Furthermore, by definition, the index $j$ is the \textit{largest} index such that $y_j > x_j$, and $k$ {is} the \textit{smallest} index {greater than $j$} such that $y_k< x_k$. It follows that  
    \begin{equation}
        \label{eq:y'_from_j+1_to_k-1}
        y'_{\ell} = y_{\ell} = x_{\ell},\hspace{1cm} \ell=j+1,\dots,k-1.
    \end{equation}
    Thus, from (\ref{eq:passouno_B}), we obtain that for each $\ell=j,\dots,k-1$, it holds that  
    \begin{align}\label{eq:passodue_B}
        \sum_{i=1}^{\ell} y'_i = \sum_{i=1}^{j-1} y'_i + y'_j  + \sum_{i=j+1}^{\ell} y'_i &\geq \sum_{i=1}^{j-1} x_i + y'_j  + \sum_{i=j+1}^{\ell} y'_i\nonumber\\
        &\geq \sum_{i=1}^{j-1} x_i + x_j + \sum_{i=j+1}^{\ell} y'_i\hspace{0.5cm} \mbox{(since $y'_j \geq x_j$)}\nonumber\\
        &{=} \sum_{i=1}^{j-1} x_i + x_j + \sum_{i=j+1}^{\ell} x_i\hspace{0.5cm}\mbox{(from (\ref{eq:y'_from_j+1_to_k-1}))}\nonumber\\
        &=\sum_{i=1}^{\ell} x_i.
    \end{align}
    From (\ref{eq:y'}), we get
    \begin{align}\label{eq:passotre_B}
        \sum_{i=1}^k y'_i &= \sum_{i=1,i\neq j,k} y_i + (y_j-\delta) + (y_k +\delta)\nonumber\\
        &= \sum_{i=1}^k y_i \geq \sum_{i=1}^k x_i\hspace{0.5cm} \mbox{(since $\x\prec \y$).}
    \end{align}
    Finally, since $\x \prec \y$, we have that
    \begin{equation}
    \label{eq:passoquattro_B}
    \begin{cases}
            \displaystyle \sum_{i=1}^\ell y'_i = \sum_{i=1}^{\ell} y_i \geq \sum_{i=1}^\ell x_i, \hspace{0.5cm} \ell=k+1,\ldots,n-1, \\
         \displaystyle\sum_{i=1}^n y'_i = \sum_{i=1}^{n} y_i = \sum_{i=1}^n x_i.    
        \end{cases}
    \end{equation}
    Therefore, from (\ref{eq:passouno_B}) and (\ref{eq:passodue_B})--(\ref{eq:passoquattro_B}), we have that $\x \prec \y'$.

    Notice that if $\delta = y_j - x_j$, then $y'_j$ is equal to $x_j$; equivalently, if $\delta = x_k - y_k$, then $y'_k$ is equal to $x_k$. Moreover, since each $\bB$-transform preserves the majorization, we can iterate the process starting from $\y'=\y\bB$. It follows that $\x$ can be derived from $\y$ by the application of a finite number of $\bB$-transforms.
    
    \smallskip
    Let us now prove the converse part of the theorem. 
    We prove it by contradiction. Hence, we assume that $\x \nprec \y$, and show that if $\x$ can be derived from $\y$ by the successive applications of a finite number of $\bB$-transforms, we get a contradiction.

    Since $\x \nprec \y$, there exists an index $\ell \in \{1,\dots,n\}$, such that
        \begin{equation}
            \sum_{i=1}^{\ell} y_i < \sum_{i=1}^{\ell} x_i.
        \end{equation}
        Moreover, by definition of $\bB$-transform (\ref{defB}), the quantity $\lambda y_j$ can be moved between two components $y_j$ and $y_k$ of $\y$, with $j > k$, only from $y_j$ to $y_k$. Therefore, the sum of the first $\ell$ components of $\y$ cannot be increased in any way through $\bB$-transforms. Consequently, this leads to a contradiction, because not all components $y_1,\dots,y_{\ell}$ can be transformed into their respective components of $\x$. Thus, it must hold that $\x \prec \y$.
\end{proof}

\begin{corollary}\label{cor:u_transforms}
    Let $\x, \y \in \mathbb{R}_+^n$. If $x\prec y$, then $\x$ can be derived from $\y$ by the application of at most $n-1$ $\bB$-transforms.
\end{corollary}
\begin{proof}
    In the proof of Theorem \ref{th:u_transforms}, we have shown 
    that the application of each $\bB$-transform 
    equalizes at least one component 
    of the intermediate vectors $\y_j=\y_{j-1}\bB_j$ 
    to a component of $\x$. Observe that since all vectors appearing in the sequence of transformation from $\y$ to $\x$ have an equal sum, the last $\bB$-transform always equalizes both the affected components to the respective components of $\x$.  As a result, $\x$ can be obtained  by the application of
     at most $n-1$ \ $\bB$-transforms.
\end{proof}

With the same technique of Lemma \ref{lemma:matrix_mul}, 
we can prove the following result.

\begin{lemma}
    \label{lemma:matrix_mul_2}
    Let $\bC$ and $\bD$ be two $n\times n$ upper triangular row-stochastic matrices. The product matrix $\bC\bD$ is still an upper triangular row-stochastic matrix.
\end{lemma}

The following theorem characterizes majorization in terms of the upper triangular row-stochastic matrix.

\begin{theorem}
    \label{th:upper_triangular}
    Let $\x, \y \in \mathbb{R}_+^n$. It holds that $\x \prec \y$ if, and only if, there exists an upper triangular row-stochastic matrix $\bU$ such that $\x = \y\bU$.
\end{theorem}
\begin{proof}
    The implication that if $\x \prec \y$, there exists an upper triangular row-stochastic matrix $\bU$ such that $\x = \y\bU$ directly derives from the results of Theorem \ref{th:u_transforms} and Lemma \ref{lemma:matrix_mul_2}. Indeed, from Theorem~\ref{th:u_transforms}, we know that $\x$ can be derived from $\y$ by the successive application of $\bB$-transforms, and from Lemma~\ref{lemma:matrix_mul_2} that the product of $\bB$-transforms is still an upper triangular row-stochastic matrix. Hence, the matrix $\bU$ obtained as the product of all $\bB$-transforms is an upper triangular row-stochastic matrix such that $\x = \y\bU$.

    \smallskip
   
    We now prove the converse implication.
  Let $\bU$ be an upper triangular row-stochastic matrix $\bU$ such that $\x = \y\bU$. 
It follows that each component of $\x$ can be written as follows:
    \begin{equation}
        \label{eq:x_j}
        x_j = \sum_{i=1}^j U_{ij}y_i,\hspace{0.5cm} j=1,\dots,n.
    \end{equation}
  By  (\ref{eq:x_j}),  we can express the sum of the first $k$ components of $\x$, with $k=1,\dots,n$, as follows:
    \begin{align*}
        x_1+\dots+x_k =& \sum_{j=1}^k \sum_{i=1}^j U_{ij}y_i\\
        =& \sum_{i=1}^k \left(\sum_{j=i}^k U_{ij}\right) y_i\\
        \leq& \sum_{i=1}^k y_i\hspace{.9cm} \mbox{(since $\sum_{j=i}^k U_{ij}\leq 1$,  given that  $U$ is row-stochastic)}.
    \end{align*}
    Thus, $\x \prec \y$.
\end{proof}

We now bound the number of nonzero elements in the matrix $\bU$ of Theorem \ref{th:upper_triangular}.
\begin{corollary}
    Let $\x,\y \in \mathbb{R_+}^n$. If $\x \prec \y$, then there exists an upper triangular row-stochastic matrix $\bU$, with at most $2n-1$ nonzero elements, such that $\x =\y\bU$.
\end{corollary}

\begin{proof}
    From Theorem \ref{th:u_transforms} and Corollary \ref{cor:u_transforms}, we know that the matrix $\bU$, for which it holds that $\x=\y\bU$, is the product of at most $n-1$ $\bB$-transforms. 
    Let 
    $$
        \bB_1,\dots,\bB_t
    $$
    be the individual matrices associated with such $t$ $\bB$-transforms, $t<n$. By  (\ref{defB}), the matrix $\bB_1$ has $n+1$ nonzero elements.

    Let $\bC$ be the matrix equal to the product of the first $m-1$ $\bB$-transforms, $m\leq t < n$, that is, $\bC = \prod_{i=1}^{m-1} \bB_i$. We show that the product
    \begin{equation}\label{eq:product_CB}
        \bC \bB_m
    \end{equation}
    gives a matrix with only one additional nonzero element with respect to $\bC$.

    Let $j, k$ be the pair of {indices} chosen to construct $\bB_m = I + \bY_m$. From (\ref{eq:product_CB}), we get
    \begin{equation}
        \label{eq:product_CB_2}
        \bC \bB_m = \bC (\bI + \bY_m) = \bC + \bC \bY_m.
    \end{equation}
    For every $\bB$-transform, we recall that we always choose the \textit{largest} index $j$ such that $y_j>x_j$, and the \textit{smallest} index $k$ greater than $j$ for which $y_k < x_k$. Therefore, all the previous $\bB$-transforms have chosen pairs of {indices} $i, \ell$ such that $i$ is greater than or at most equal to $j$. Consequently, in the matrix $\bC$, all the rows above the $j$-th row are equal to the rows of the identity matrix. In addition, by the definition, the matrix $\bY$ has nonzero elements only in the $j$-th row, in positions $Y_{jj}$ and $Y_{jk}$, respectively. Hence, the matrix $\bC\bY_m$ has nonzero elements only in the entries $(CY_m)_{jj}$ and $(CY_m)_{jk}$. Since the element $(CY_m)_{jj}$ will be added to the diagonal element of $\bC$, the \textit{only} new nonzero element is $(CY_m)_{jk}$.
    Hence, from~(\ref{eq:product_CB_2}), we get that the product gives a matrix with only one new additional element with respect to $\bC$.

    Since each product generates a matrix with only one additional nonzero element with respect to the previous one, we obtain that the final matrix $\bU$ has at most $n+1+n-2=2n-1$ nonzero elements.
\end{proof}

We summarize our results in the next theorem, in the fashion of the classic Theorem \ref{th:equivalence_T}.
\begin{theorem}
    If $\x, \y \in \mathbb{R}_+^n$, the following conditions are equivalent:
    \begin{enumerate}
        \item[(1)] $\x \prec \y$;
        \item[(2)] $\x = \y\bU$ for some upper triangular row-stochastic matrix $\bU$;
        \item[(3)] $\x$ can be derived from $\y$ by the successive applications of at most $n-1$ $\bB$-transforms, as defined in (\ref{u_transform}).
    \end{enumerate}
\end{theorem}
\begin{proof}
    The equivalences are a direct consequence of the results of Theorem \ref{th:u_transforms}, Corollary \ref{cor:u_transforms}, and Theorem \ref{th:upper_triangular}.
\end{proof}

Let us now see an example of the construction of the matrix $\bU$.

\begin{example}
Let $\y = [0.6, 0.2, 0.1, 0.1]$ and $\x = [0.3, 0.3, 0.3, 0.1]$ with $\x \prec \y$.

The first $\bB$-transform modifies  the elements $y_1$ and $y_2$ with $\delta = \min{\{}y_1-x_1,x_2-y_2{\}} = 0.1$ and $\lambda=\delta/y_1=1/6$. As a result, the associated matrix $\bB_1$ is the following one:
$$
    \bB_1 = \begin{bmatrix}
            1 & 0 & 0 & 0\\
            0 & 1 & 0 & 0\\
            0 & 0 & 1 & 0\\
            0 & 0 & 0 & 1
        \end{bmatrix} +
        \begin{bmatrix}
            -\frac{1}{6} & \frac{1}{6} & 0 & 0\\
            0 & 1 & 0 & 0\\
            0 & 0 & 1 & 0\\
            0 & 0 & 0 & 1
        \end{bmatrix} =
        \begin{bmatrix}
            \frac{5}{6} & \frac{1}{6} & 0 & 0\\
            0 & 1 & 0 & 0\\
            0 & 0 & 1 & 0\\
            0 & 0 & 0 & 1
        \end{bmatrix}
$$
and $\y^1 {= \y \bB_1} = [0.5, 0.3, 0.1, 0.1]$.

The second $\bB$-transform affects  the elements $y^1_1$ and $y^1_3$ with $\delta= \min{\{}y^1_1-x_1,x_3-y^1_3{\}} = 0.2$ and $\lambda=\delta/y^1_1=2/5$. Hence,
the matrix $\bB_2$ is the following one
$$
    \bB_2 = \begin{bmatrix}
            1 & 0 & 0 & 0\\
            0 & 1 & 0 & 0\\
            0 & 0 & 1 & 0\\
            0 & 0 & 0 & 1
        \end{bmatrix} +
        \begin{bmatrix}
            -\frac{2}{5} & 0 & \frac{2}{5} & 0\\
            0 & 1 & 0 & 0\\
            0 & 0 & 1 & 0\\
            0 & 0 & 0 & 1
        \end{bmatrix} =
        \begin{bmatrix}
            \frac{3}{5} & 0 & \frac{2}{5} & 0\\
            0 & 1 & 0 & 0\\
            0 & 0 & 1 & 0\\
            0 & 0 & 0 & 1
        \end{bmatrix}
$$
and $\y^2 {= \y^1 \bB_2}= [0.3, 0.3, 0.3, 0.1] = \x$. Hence, the final matrix is:
\begin{equation*}
    \begin{aligned}
        \bU = \bB_1 \bB_2 =
        \begin{bmatrix}
            \frac{1}{2} & \frac{1}{6} & \frac{2}{6} & 0\\
            0 & 1 & 0 & 0\\
            0 & 0 & 1 & 0\\
            0 & 0 & 0 & 1
        \end{bmatrix}
    \end{aligned}
\end{equation*}
It is worth pointing out that the above example also shows that
the matrices $\bU$ of this section are not simply the inverses
of the matrices  $\bL$ of Section \ref{sec:L}.

\end{example}


\section{Applications}\label{appl}
We recall that a real-valued function $\phi:A\subseteq \mathbb{R}^n\longmapsto \mathbb{R}$ 
is said to be \textit{Schur-concave} 
\cite{marshall1979inequalities}  if
$$
    \x \prec \y  \implies \phi(\x) \geq \phi(\y).
$$
In the rest of this section, we will assume that the set $A$ corresponds to the
$(n-1)$-dimensional probability simplex $\cP_n$, defined as
$$\cP_n=\{\x=(x_1, \ldots , x_n) \mid x_1 \geq \dots \geq x_n >0 \mbox{ and }
\sum_{i=1}^nx_i=1\}.$$
It is well known that the Shannon entropy  $H(\x) = - \sum_{i=1}^n x_i \log_2 x_i$, is  Schur-concave  over $\cP_n$. Therefore,
for any  $\x, \y\in \cP_n$ , if $\x \prec \y$, it holds that
\begin{equation}\label{eq:simple_property}
    H(\x) \geq H(\y).
\end{equation}
The above inequality (\ref{eq:simple_property}) is widely used in information theory.
There are several improvements to the basic inequality (\ref{eq:simple_property}). 
{For instance,
the authors of \cite{Ho} proved that 
for any  $\x, \y\in \cP_n$ , if $\x \prec \y$, then it holds that
\begin{equation}\label{eq:HV}
    H(\x) \geq H(\y)+ D(\y||\x),
    \end{equation}
    where $D(\y||\x)=\sum_iy_i\log (y_i/x_i)$ is the relative entropy between $\y$ and $\x$.}

The paper \cite{cohen1993majorization} proved a different
strengthening of (\ref{eq:simple_property})  (see also Proposition A.7.e.  of \cite{marshall1979inequalities}). More precisely, 
Proposition A.7.e.  of \cite{marshall1979inequalities} states that 
if $\x, \y\in \cP_n$ and 
  $\x \prec \y$, then  it holds that
\begin{equation}\label{simpleimproved}
    H(\x) \geq \alpha(\bP)\log_2 n + (1-\alpha(\bP))H(\y) \geq H(\y),
\end{equation}
where $\bP=[P_{ij}]$ is a doubly stochastic matrix for which $\x = \y\bP$, and $\alpha(\bP)$ is the Dobrushin coefficient of ergodicity of $\bP$, defined as:
\begin{align*}
    \alpha(\bP) = \min_{\ell,m} \sum_{i=1}^n \min{\{}P_{\ell i}, P_{m i}{\}}.
\end{align*}
{It might be useful to 
recall that there are papers that intend the Dobrushin coefficient of ergodicity of $\bP$
as $1-\alpha(\bP)$.}

We show that our results from the previous sections can be used to obtain a different improvement
of the basic inequality (\ref{eq:simple_property}).
In fact, we prove the following result.

\begin{theorem}\label{th:application}
    Let  $\x, \y\in \cP_n$ and 
  $\x \prec \y$. Moreover, let 
  $\bU$ be the upper triangular matrix obtained through the sequence of $\bB$-transforms described in Theorem $\ref{th:u_transforms}$
  for which $\x = \y \bU$. If $\x \neq \y$, it holds that
    \begin{align}
        H(\x) &\geq (1-\alpha(\bU)) H(\y) + \sum_{i=1}^n x_i \log_2 \frac{1}{(\bu \bU)_i} -(1-\alpha(\bU)) \log_2 n\\
        &> \alpha(\bU) \log_2 n + (1-\alpha(\bU))H(\y),
    \end{align}
    where $\bu = (1/n,\dots,1/n)$ and $\alpha(\bU)$ is the 
    Dobrushin coefficient of ergodicity of $\bU$.
\end{theorem}
To prove the theorem, we need some intermediate results.
We recall that a matrix $\bC$ is said to be column-allowable if each column contains at least one positive element \cite{cohen1993majorization}.
From \cite{cohen1993majorization}[Thm. 1.4], we obtain the following Lemma\footnote{Theorem 1.4 of  \cite{cohen1993majorization} is for general $f$-divergences.}.
\begin{lemma}\label{lemma:relative_entropy}
    Let $\bC\in \mathbb{R}_+^{n \times n}$ be a row-stochastic and column-allowable matrix and let $\x, \y\in \cP_n$, then
    \begin{equation}
        D(\x \bC \Vert \y \bC) \leq (1-\alpha(\bC))D(\x \Vert \y),
    \end{equation}
    where $D(\x \Vert \y) = \sum_{i=1}^n x_i \log _2(x_i/y_i)$ is the relative entropy between $\x$ and $\y$.
\end{lemma}
We note that Lemma \ref{lemma:relative_entropy} is a classical example of
the role of contraction coefficients. Contraction coefficients
are important 
quantities in strong data-processing inequalities \cite{Sa}.

By exploiting the knowledge of the structure 
of the matrices $\bU$  of Section \ref{sec:upper}, we obtain the following result.

\begin{lemma}\label{lemma:>log_n}
    Let $\x, \y\in \cP_n$, such that   $\x \prec \y$.
    Moreover,  let $\bU$ be the upper triangular 
    matrix obtained through the sequence of $\bB$-transforms described in Theorem $\ref{th:u_transforms}$ such that $\x = \y \bU$. If $\x \neq \y$, it holds that
    \begin{equation}
        \sum_{i=1}^n x_i \log_2 \frac{1}{(\bu \bU)_i} > \log_2 n.
    \end{equation}
\end{lemma}
\begin{proof}
    Let $\bU = \prod_{i=1}^s \bB^i$ be the upper triangular matrix obtained as product of the $s$ $\bB$-transforms $\bB^1,\dots,\bB^s$. To prove the lemma, we first show that 
    \begin{equation}\label{eq:primo_risultato}
        \sum_{i=1}^n x_i \log_2 \frac{1}{(\bu \bB^1)_i} > \log_2 n.
    \end{equation}
   Successively, we prove that for each $\ell=2,\dots,s-1$, it holds that
    \begin{equation}\label{eq:secondo_risultato}
        \sum_{i=1}^n x_i \log_2 \frac{1}{(\bu (\bC \bB^{\ell}))_i} -  \sum_{i=1}^n x_i \log_2 \frac{1}{(\bu \bC)_i} = \sum_{i=1}^n x_i \log_2 \frac{(\bu \bC)_i}{(\bu (\bC \bB^{\ell}))_i}> 0,
    \end{equation}
    where $\bC = \prod_{i=1}^{\ell-1} \bB^i$.
    \smallskip

    Let $j, k$ be the pair of {indices}, with $j < k$ chosen in the first $\bB$-transform $\bB^1$, then we know that 
    \begin{equation}\label{eq:def_B^1}
        B^1_{\ell m} = 
        \begin{cases}
            1-\lambda &\qquad\mbox{ if } \ell=m=j,\\
            \lambda &\qquad\mbox{ if } \ell=j \mbox{ and } m=k,\\
            1 &\qquad\mbox{ if } \ell=m\neq j\\
            0 &\qquad\mbox{ otherwise. }
        \end{cases}
    \end{equation}
    Hence, from (\ref{eq:def_B^1}) and by noticing that 
    for each $i=1, \ldots , n$, it holds that 
    $(\bu \bB^1)_i = (1/n)\sum_{\ell = 1}^n B^1_{\ell i}$,  we get the following series of equalities and inequalities:
    \begin{align}
        \sum_{i=1}^n x_i \log_2 \frac{1}{(\bu \bB^1)_i} - \log_2 n &= \sum_{i=1}^n x_i \log_2 \frac{1}{(\bu \bB^1)_i} - \sum_{i=1}^n x_i \log_2 n\nonumber\\
        &= \sum_{i=1}^n x_i \log_2 \frac{1}{n(\bu \bB^1)_i}\nonumber\\
        &= \sum_{i=1}^n x_i \log_2 \frac{1}{n \frac{1}{n}\sum_{\ell = 1}^n B^1_{\ell i}}\nonumber\\
        &=\sum_{i=1}^n x_i \log_2 \frac{1}{\sum_{\ell = 1}^n B^1_{\ell i}}\nonumber\\
        &=\sum_{i\neq j, k} x_i \log_2 \frac{1}{B^1_{ii}} + x_j \log_2 \frac{1}{B^1_{jj}}\nonumber\\
         &\ \ \ \ \ + x_k \log_2 \frac{1}{B^1_{kk} + B^1_{jk}}\hspace{0.2cm}\mbox{(from (\ref{eq:def_B^1}))}\nonumber\\
        &=x_j \log_2 \frac{1}{1-\lambda} + x_k \log_2 \frac{1}{1+\lambda}\nonumber\\
        &\geq x_k \log_2 \frac{1}{(1-\lambda)(1+\lambda)}\hspace{2cm}\mbox{(since $x_j \geq x_k$)}\nonumber\\
        &= x_k \log_2 \frac{1}{1-\lambda^2} > 0.
    \end{align}
    Thus, we have proved that (\ref{eq:primo_risultato}) holds.

\smallskip

Let $\bC = \prod_{i=1}^{\ell-1} \bB^i$ for an arbitrary $\ell \in \{2, \dots, s-1\}$. Let $j, k$ be the pair of {indices} with $j < k$ chosen in the successive $\bB$-transform $\bB^{\ell} = \bI + \bY^{\ell}$, then we know that 
\begin{equation}\label{eq:def_Y^l}
        Y^{\ell}_{t m} =\begin{cases}
            -\lambda &\qquad\mbox{ if } t=m=j,\\
            \lambda &\qquad\mbox{ if } t=j \mbox{ and } m=k,\\
            0 &\qquad\mbox{ otherwise. }
        \end{cases}
\end{equation}
Therefore, we get
\begin{equation}\label{eq:product_}
    \bC \bB^{\ell} = \bC (\bI + \bY^{\ell}) = \bC + \bC \bY^{\ell}.
\end{equation}

For every $\bB$-transform, we recall that we always choose the \textit{largest} index $j$ such that $y_j>x_j$, and the \textit{smallest} index $k$ greater than $j$ for which $y_k < x_k$. Therefore, all the previous $\bB$-transforms have chosen pairs of {indices} $t, m$ such that $t$ is greater than or at most equal to $j$. Consequently, in the matrix $\bC$, all the rows above the $j$-th row are equal to the rows of the identity matrix. 
Thus, from (\ref{eq:def_Y^l}), it follows that the only nonzero entries in the matrix $\bC \bY^{\ell}$ are $(C Y^{\ell})_{jk} = \lambda C_{jj}$ and $(C Y^{\ell})_{jj} = -\lambda C_{jj}$. Therefore, from (\ref{eq:product_}), the only entries of $\bC$ that are 
different from the corresponding entries of 
$\bC \bB^{\ell}$ are $C_{jj}$ and $C_{jk}$. In particular, in the matrix 
$\bC \bB^{\ell}$, the entry 
$C_{jj}$ of $\bC$ is decremented by $\lambda C_{jj}$ and $C_{jk}$ is incremented by $\lambda C_{jj}$. Therefore, we have
\begin{equation}\label{eq:CB_l_def}
    (CB^{\ell})_{tm} = 
    \begin{cases}
            C_{jj}-\lambda C_{jj} &\qquad\mbox{ if } t=m=j,\\
            C_{jk} + \lambda C_{jj} &\qquad\mbox{ if } t=j \mbox{ and } m=k,\\
            C_{tm} &\qquad\mbox{ otherwise. }
        \end{cases}
\end{equation}
Thus, we get 
\begin{align}\label{eq:tmp_step}
    \sum_{i=1}^n x_i \log_2 \frac{(\bu \bC)_i}{(\bu (\bC \bB^{\ell}))_i}
    =& \sum_{i=1}^n x_i \log_2 \frac{(1/n)\sum_{t=1}^n C_{ti}}{(1/n) \sum_{t=1}^ n (CB^{\ell})_{ti}}\nonumber\\
    =&\sum_{i=1}^n x_i \log_2 \frac{\sum_{t=1}^n C_{ti}}{\sum_{t=1}^ n (CB^{\ell})_{ti}}\nonumber\\
    =& \sum_{i\neq j, k} x_i \log_2 \frac{\sum_{t=1}^n C_{ti}}{\sum_{t=1}^n C_{ti}}  +x_j \log_2 \frac{\sum_{t=1}^n C_{tj}}{\sum_{t=1}^n C_{tj} - \lambda C_{jj}}\nonumber\\
    &+ x_k \log_2 \frac{\sum_{t=1}^n C_{tk}}{\sum_{t=1}^n C_{tk} + \lambda C_{jj}}\nonumber\hspace{1.5cm}\mbox{(from (\ref{eq:CB_l_def}))}\\
    =& x_j \log_2 \frac{\sum_{t=1}^n C_{tj}}{\sum_{t=1}^n C_{tj} - \lambda C_{jj}} + x_k \log_2 \frac{\sum_{t=1}^n C_{tk}}{\sum_{t=1}^n C_{tk} + \lambda C_{jj}}.
\end{align}
Observe that by construction of $\bC$, it holds that $C_{jj}\leq 1$ and that the only nonzero element in the $j$-th column is $C_{jj}$. Thus, it follows that $\sum_{t=1}^n C_{tj} = C_{jj}$. Similarly, it also holds that $C_{kk}=1$. In fact, since $y_k < x_k$, the $\bB$-transforms 
do not modify the $k$-th row (that remains equal to the 
row of the identity matrix). Thus, it follows that $\sum_{t=1}^n C_{tk} \geq 1$.
Therefore, we can rewrite (\ref{eq:tmp_step}) as follows
\begin{align}
    \sum_{i=1}^n x_i \log_2 \frac{(\bu \bC)_i}{(\bu (\bC \bB^{\ell}))_i} &= x_j \log_2 \frac{C_{jj}}{C_{jj} - \lambda C_{jj}} + x_k \log_2 \frac{\sum_{t=1}^n C_{tk}}{\sum_{t=1}^n C_{tk} + \lambda C_{jj}}\nonumber\\
    &= x_j \log_2 \frac{1}{1-\lambda} + x_k \log_2 \frac{1}{1+ \frac{\lambda C_{jj}}{\sum_{t=1}^n C_{tk}}}\nonumber\\
    &\geq x_k \log_2 \frac{1}{(1-\lambda)(1+ \frac{\lambda C_{jj}}{\sum_{t=1}^n C_{tk}})} \hspace{1.5cm} \mbox{(since $x_j \geq x_k$)}\nonumber\\
    &\geq x_k \log_2 \frac{1}{(1-\lambda)(1+\lambda)}\hspace{1.5cm} 
    \left (\mbox{since $\displaystyle 1+ \lambda \geq 1+ \frac{\lambda C_{jj}}{\sum_{t=1}^n C_{tk}}$} \right )\nonumber\\
    &> 0.
\end{align}
Thus, we proved that (\ref{eq:secondo_risultato}) holds.

Since both (\ref{eq:primo_risultato}) and (\ref{eq:secondo_risultato}) hold, we have proved the Lemma, given that
$$
    \sum_{i=1}^n x_i \log_2 \frac{1}{(\bu \bU)_i} > \sum_{i=1}^n x_i \log_2 \frac{1}{(\bu \bB^1)_i} >\log_2 n.
$$
\end{proof}

We can now prove Theorem \ref{th:application}.
\begin{proof}
    Observe that the matrix $\bU$ is a row-stochastic and column-allowable matrix. Thus, from Lemma~(\ref{lemma:relative_entropy}), we 
     obtain
    
    \begin{equation*}
        D(\y\bU \Vert \bu \bU) \leq (1-\alpha(\bU))D(\y \Vert \bu),
    \end{equation*}
    and, since $\x = \y\bU$,
    \begin{equation}\label{eq:nuova}
       D(\x \Vert \bu \bU)\leq (1-\alpha(\bU))D(\y \Vert \bu).
    \end{equation}
    Expanding the divergences in (\ref{eq:nuova}), we get
    \begin{equation}\label{eq:lemma_application}
        \sum_{i=1}^n x_i \log_2 \frac{1}{(\bu \bU)_i}-H(\x) \leq (1-\alpha(\bU))(\log_2 n - H(\y)).
    \end{equation}  
    From (\ref{eq:lemma_application}), we get the following lower bound on the entropy of $\x$:
    \begin{align}\label{eq:tmp_result}
        H(\x) \geq (1-\alpha(\bU))H(\y) + \sum_{i=1}^n x_i \log_2 \frac{1}{(\bu \bU)_i} - (1-\alpha(\bU))\log_2 n.
    \end{align}
    From Lemma \ref{lemma:>log_n}, we know that
    \begin{equation}\label{eq:lemma_recall}
        \sum_{i=1}^n x_i \log_2 \frac{1}{(\bu \bU)_i} > \log_2 n.
    \end{equation}
    Thus, by applying (\ref{eq:lemma_recall}) to (\ref{eq:tmp_result}), we finally get
    \begin{align*}
        H(\x) &\geq (1-\alpha(\bU))H(\y) + \sum_{i=1}^n x_i \log_2 \frac{1}{(\bu \bU)_i} - (1-\alpha(\bU))\log_2 n\\
        &> (1-\alpha(\bU))H(\y) + \log_2 n - (1-\alpha(\bU))\log_2 n\\
        &= (1-\alpha(\bU))H(\y) + \alpha(\bU))\log_2 n\\
        &\geq H(\y) \hspace{5cm}\mbox{(since $H(\y) \leq \log_2 n$).}
    \end{align*}\end{proof}

\section{Conclusions}
In this paper, we have introduced two novel  characterizations of
the classical concept of majorization in terms of upper triangular (resp., 
lower triangular) row-stochastic matrices.
The interesting features of our upper triangular (resp., 
lower triangular) row-stochastic matrices are that they are quite sparse 
in the sense that they have few nonzero elements this 
property might be useful in 
practical applications. Finally, we have used our new characterization of 
 majorization in terms of upper triangular row-stochastic matrices
to derive an improved entropy inequality.  
We mention that one could derive a similar (albeit \textit{non}equivalent)
improved entropy inequality by using the characterization of 
 majorization in terms of lower triangular row-stochastic matrices that 
 we have given in Section \ref{sec:L}. To do so, the way to proceed is similar 
 to the one we presented in 
 Section \ref{appl}.

\section*{Acknowledgments}
The authors would like to thank the referees and the Guest Editor Professor I. Sason
for their corrections and useful suggestions. 

This work was partially supported by project SERICS (PE00000014) under the NRRP MUR program funded by the EU-NGEU.

\end{document}